\documentclass[11pt]{article}

\usepackage{amssymb,amsmath,amsthm,amsfonts,fullpage}
\usepackage[pdftex]{graphicx}

\usepackage[mathscr]{eucal}

\newtheorem{theorem}{Theorem}

\def\B{\mathscr{B}}
\def\P{\mathcal{P}}
\def\L{\mathcal{L}}
\def\A{\mathcal{A}}
\def\Y{\mathscr{Y}}

\def\R{\mathbb{R}}

\def\Z{\mathbb{Z}}

\begin{document}
\title{Approximation with one-bit polynomials in Bernstein form}
\date{\today}
\author{C. S{\.\i}nan G\"unt\"urk\footnote{NYU Courant Institute, email: \tt{gunturk@cims.nyu.edu}.}~ and Weilin Li\footnote{CUNY City College and NYU Courant Institute, email: \tt{wli6@ccny.cuny.edu},  \tt{weilinli@cims.nyu.edu}.} 
} 
\maketitle
\begin{center}
	{\em Dedicated to Professor Ron DeVore on the occasion of his 80th birthday.}
\end{center}

\begin{abstract}
We prove various theorems on approximation using polynomials with integer coefficients in the Bernstein basis of any given order. In the extreme, we draw the coefficients from $\{ \pm 1\}$ only. A basic case of our results states that for any Lipschitz function $f:[0,1] \to [-1,1]$ and for any positive integer $n$, there are signs $\sigma_0,\dots,\sigma_n \in \{\pm 1\}$ such that
$$\left |f(x) - \sum_{k=0}^n \sigma_k \, \binom{n}{k} x^k (1-x)^{n-k} \right | \leq \frac{C (1+|f|_{\mathrm{Lip}})}{1+\sqrt{nx(1-x)}} ~\mbox{ for all } x \in [0,1].$$  
More generally, we show that higher accuracy is achievable for smoother functions: For any integer $s\geq 1$, if $f$ has a Lipschitz $(s{-}1)$st derivative, then approximation accuracy of order $O(n^{-s/2})$ is achievable with coefficients in $\{\pm 1\}$ provided $\|f \|_\infty < 1$, and of order $O(n^{-s})$ with unrestricted integer coefficients, both uniformly on closed subintervals of $(0,1)$ as above. Hence these polynomial approximations are not constrained by the saturation of classical Bernstein polynomials. Our approximations are constructive and can be implemented using feedforward neural networks whose weights are chosen from $\{\pm 1\}$ only. 
\end{abstract}


\section{Introduction}

It is a classical result that a continuous function $f:[0,1] \to \R$ can be approximated uniformly by polynomials with integer coefficients if and only if $f(0)$ and $f(1)$ are integers. Here the integer coefficients are understood to be with respect to the default ``power basis,'' i.e. the mononomials $\{1,x,x^2,\dots\}$. The necessity of the condition is immediate. The sufficiency, on the other hand, is non-obvious at best, yet the following constructive proof by Kantorovich \cite{kantorovich1931} (see also \cite[Ch.2.4]{lorentz2}), is remarkably short and transparent:
Recall that the Bernstein polynomial of $f$ of order $n$ (and degree $\leq n$), defined by 
$$
 B_n(f,x):= B_n(f)(x):= \sum_{k=0}^n f\Big (\frac{k}{n} \Big ) \binom{n}{k} x^k (1-x)^{n-k},
$$
converges to $f$ uniformly. Set
\begin{equation}\label{quantized0}
 B_n^*(f,x):= B^*_n(f)(x) := \sum_{k=0}^n \left [ f\Big (\frac{k}{n}\Big ) \binom{n}{k} \right ] x^k (1-x)^{n-k},
\end{equation}
where $[u]\in \Z$ stands for any rounding of $u \in \R$ to an immediate neighboring integer value. It is evident that $B_n^*(f,x)$ is a polynomial with integer coefficients, and of degree at most $n$. With the assumption that $f(0)$ and $f(1)$ are integers, the total rounding error can be bounded uniformly over $x \in [0,1]$ via
\begin{equation}\label{rounding1}
 |B_n(f,x) - B_n^*(f,x)| \leq \sum_{k=1}^{n-1}  x^k (1-x)^{n-k} 
 \leq \frac{1}{n} \sum_{k=1}^{n-1} \binom{n}{k} x^k (1-x)^{n-k} < \frac{1}{n}
\end{equation}
which shows $B^*_n(f) \to f$ uniformly as well. 

Since accuracy of the Bernstein polynomial approximation saturates at the rate $1/n$ (unless $f$ is a linear polynomial), the rate at which $B_n^*(f)$ converges to $f$ is as good as that of $B_n(f)$. Kantorovich actually proved a stronger result in \cite{kantorovich1931}, showing that the error of best approximation to $f$ by polynomials of degree $n$ with integer coefficients is bounded by $2 E_n(f) + 1/n$ where $E_n(f)$ denotes the error of best approximation of $f$ by unconstrained polynomials of degree $n$. This can be shown by employing the polynomial of best approximation as a surrogate instead of the Bernstein polynomial (see \cite{lorentz1, ferguson1980}).

The history of approximation by polynomials with integer coefficients is rich, with the earliest result going back to P{\'a}l \cite{pal1914}, followed shortly by Kakeya \cite{kakeya1914} and Chlodovsky \cite{chlodovsky1925}. For an extensive treatment of the subject, we refer to Ferguson's classical text \cite{ferguson1980}. One of the important characteristics of the theory is that uniform approximation of continuous functions is only possible on intervals of length less than $4$, and then only with certain arithmetic constraints on $f$: We already saw that on $[0,1]$ it is necessary (and sufficient) that $f(0)$ and $f(1)$ are integers. On $[-\alpha,\alpha]$ where $\alpha < 1$, it is necessary (and sufficient) that $f(0)$ is an integer, whereas on $[-1,1]$, it is necessary (and sufficient) that $f(-1),f(0),f(1),\frac{1}{2}(f(-1){+}f(1))$ are all integers. The number of such arithmetic conditions increases without bound as the length of the interval increases towards $4$. In the other extreme, there are no arithmetic conditions when approximation is sought on closed intervals containing no integers. In the context of this paper we will assume that all uniform approximation takes place on a subinterval $[a,b] \subset (0,1)$. 

It is known that Kantorovich's method, which yields a residual error of $1/n$, is suboptimal. Indeed, for any $0<a<b<1$, the error of best uniform approximation to any $f\in C([a,b])$ by a degree $n$ polynomial with integer coefficients is bounded by $E_n(f) + \rho^n$, where $\rho:=\rho(a,b) < 1$ (see \cite{trigub1962}, \cite[Thm 11.15]{ferguson1980}). Consequently, in all traditional smoothness classes of finitely many derivatives, the constraint of integer coefficients does not hamper the rate at which these functions can be approximated uniformly by polynomials. It is, however, important to remember the constraint on the domain $[a,b]$.

This paper is concerned with approximations by polynomials with integer coefficients subject to additional, stringent conditions on what these integers can be, while still guaranteeing similar approximation properties. To explain what these stringent conditions are, we first need some additional notation:
By the Bernstein basis of order $n$, we refer to the list of polynomials $\B_n:=(p_{n,k})_{k=0}^n$
where
\begin{equation}\label{Bernsteinbasis}
p_{n,k}(x) := \binom{n}{k}x^k (1-x)^{n-k}.
\end{equation}
$\B_n$ is a basis of the vector space $\P_n$ of polynomials (of a real variable) of degree at most $n$.
Let us denote by $\B^\circ_n: = (p^\circ_{n,k})_{k=0}^n$ the plain (unnormalized) version of this basis, given by $p^\circ_{n,k}(x):= x^k(1-x)^{n-k}$. Let us also denote by $\Pi_n:=(\pi_k)_0^n$ the power basis of degree $n$ given by $\pi_k(x):=x^k$. Then $B_n^*(f)$, as defined in \eqref{quantized0}, can be viewed as an approximation to $f$ from the lattice $\L(\B^\circ_n)$ generated by $\B^\circ_n$, and of course, also from $\L(\Pi_n)$ as originally intended. In fact, $\L(\B^\circ_n) = \L(\Pi_n)$ since for each $k=0,\dots,n$, we have $p^\circ_{n,k} \in \L(\Pi_n)$ and $\pi_k \in \L(\B^\circ_n)$. The first claim is immediate, and the latter is seen by noting that
$$ x^k = x^k(x + (1-x))^{n-k} = \sum_{j=0}^{n-k} \binom{n-k}{j}x^{k+j}(1-x)^{n-k-j} = \sum_{l=k}^n \binom{n-k}{l-k} p^\circ_{n,l}(x).$$

Meanwhile, $\L(\B_n)$, which we will refer to as the {\em Bernstein lattice}, is a significantly smaller sublattice of $\L(\B^\circ_n)$ as $n$ increases, and therefore, approximation by its elements presents an increasingly coarser rounding (quantization) problem. Indeed, the fundamental cell of $\L(\B_n)$ contains $M_n:=\prod_{k=0}^n \binom{n}{k}$ elements of $\L(\B^\circ_n)$, and therefore, encoding the elements of $\L(\B^\circ_n)$ requires $\mu_n:=\frac{1}{n+1} \log_2 M_n$ times as many bits per basis polynomial as it would require for $\L(\B_n)$. It can be checked that $M_n$ grows as $\exp(cn^2)$, therefore $\mu_n$ grows linearly. 

To motivate the same point from an approximation perspective, let us inspect what happens to the total rounding error bound when we gradually coarsen the lattice $\L(\B^\circ_n)$ towards $\L(\B_n)$ while employing the same simple rounding algorithm: For any $\alpha \in [0,1]$, let $\Delta_{n,k}:=\Delta_{n,k}(\alpha)$ be the integer part of $\binom{n}{k}^{\alpha}$, and 
analogous to \eqref{quantized0}, consider 
the rounding of $f$ to the lattice generated by $\B^\alpha_n:=(\Delta_{n,k}p^\circ_{n,k})_{k=0}^n$ by means of the polynomial
\begin{equation}\label{quantized1}
B_n^{*,\alpha}(f,x):=\sum_{k=0}^n \left [ f\Big (\frac{k}{n}\Big ) \binom{n}{k} \Delta_{n,k}^{-1} \right ] \Delta_{n,k} x^k (1-x)^{n-k}.
\end{equation}
Noting that $\Delta_{n,0}=\Delta_{n,n}=1$, and again assuming that $f(0)$ and $f(1)$ are integers, the total rounding error is now bounded above by
\begin{eqnarray}
 \sum_{k=1}^{n-1}  \binom{n}{k}^{\alpha} x^k (1-x)^{n-k} 
 & = & \sum_{k=1}^{n-1} (p_{n,k}(x))^\alpha (p^\circ_{n,k}(x))^{1-\alpha} \nonumber \\
 & \leq & \left (\sum_{k=1}^{n-1} p_{n,k}(x) \right )^{\alpha}  \left (\sum_{k=1}^{n-1} x^k(1-x)^{n-k} \right )^{1-\alpha}
  \nonumber \\
 & < & n^{-1+\alpha},\nonumber
\end{eqnarray}
where in the last step we have used the findings of \eqref{rounding1}. 

This simple extension shows that it is possible to approximate from $\L(\B^\alpha_n)$ without effort for all $\alpha <1$, but the method breaks down at $\alpha=1$, i.e. for the Bernstein lattice $\L(\B_n)$. In this paper, our primary focus will be on enabling approximation from subsets of this lattice by means of a more advanced quantization method known as {\em noise-shaping quantization}.

To explain what this means, let us recall that each $p_{n,k}$ is a bump function that peaks at $k/n$. However, the Bernstein basis is poorly localized as a whole. As can be seen by the Laplace (normal) approximation to the binomial distribution, each $p_{n,k}$ spreads significantly over the neighboring $O(\sqrt{n})$ basis functions, making $\B_n$ behave approximately like a frame of redundancy $O(\sqrt{n})$. This heuristic suggests that there is numerical flexibility in the choice of coefficients when functions are approximated by linear combinations of the $p_{n,k}$; this flexibility then leads to the possibility of coarse quantization. 

In general, noise-shaping quantization refers to the principle of arranging the quantization noise (i.e. the quantization error of the coefficients) to be mostly invisible to the accompanying reconstruction operator. Also known as ``one-bit'' quantization due to its potential to enable very coarse quantization, noise-shaping quantization was first introduced for analog-to-digital conversion circuits during the 60s (\cite{IY,IYM62}) where it became established as Sigma-Delta modulation (or $\Sigma\Delta$ quantization), though some of its ideas can be found in the theory of Beatty sequences. $\Sigma\Delta$ quantization gained some popularity in information theory through the works of Gray et al (e.g. \cite{grayquantization}), but remained largely unknown in the general mathematics community until the groundbreaking work of Daubechies and DeVore \cite{DD}. Since then, the theory has been extended significantly to apply to various approximation problems with quantized coefficients. We refer the reader to \cite{exp_decay,BPY1,mathofAtoD,noiseshaping} for some of the recent mathematical evolution of the subject, and to \cite{Candy-Temes,NST96,SchTe04} for engineering applications.

In this paper, we will use methods of noise-shaping quantization to establish various results concerning approximation from the Bernstein lattices, including the following:
\begin{itemize}
 \item  The collection $\displaystyle \bigcup_{n=1}^\infty \L(\B_n)$ of all Bernstein lattices is dense in $L^p([0,1])$ for all $1 \leq p < \infty$, and in $C([a,b])$ for any $0< a < b < 1$. More precisely, in each of these spaces, the distance from $\L(\B_n)$ to any $f$ goes to $0$.
 
 \item For any $0< a < b < 1$, if $f \in C([a,b])$ can be approximated uniformly to within $O(n^{-s})$ by polynomials of degree $n$ with real coefficients, then it can also be approximated uniformly to within $O(n^{-s})$ by elements of $\L(\B_n)$. 

 \item Any $f \in \mathrm{Lip}([0,1])$ such that $\|f\|_\infty \leq 1$ can be approximated to within $O(n^{-1/2})$ by a polynomial with coefficients in $\{\pm 1\}$ in the Bernstein basis of order $n$.  The approximation rate improves to $O(n^{-s/2})$  if $f^{(s-1)} \in \mathrm{Lip}([0,1])$, provided $\|f\|_\infty < 1$. These approximations are uniform over any given $[a,b] \subset (0,1)$.
\end{itemize}

The paper is organized as follows: Section \ref{generalmethod} describes how the noise-shaping method of $\Sigma\Delta$ quantization works in connection with the Bernstein basis. Specific error bounds are given in Section \ref{specificbounds} where both $L^p$ spaces and smoothness classes are considered. This section also describes how iterated Bernstein operators enable higher order $\Sigma\Delta$ quantizers.

An unexpected and novel contribution of this paper is the effective use of noise-shaping quantization methods in the setting of {\em linearly independent} systems of vectors. To our knowledge, the Bernstein system constitutes the first example of this kind. In Appendix A, we quantify the sense in which the Bernstein basis $\B_n$ behaves like a frame of redundancy $\sqrt{n}$ by deriving the exact eigenvalue distribution of the associated frame operator. 

Another novel contribution of this work is computational: our one-bit polynomial approximations can be computed exactly by means of feedforward neural networks whose weights are chosen from $\{\pm 1\}$ only. We show how this is done in Appendix B.

\section{Approximation by quantized polynomials in Bernstein form}
\label{generalmethod}
All of our approximations via quantized polynomials in Bernstein form will follow a two-stage process. The first stage consists of classical polynomial approximation and the second stage is quantization through noise-shaping. For any given function $f:[0,1] \to \R$ in a suitable function class, we will first approximate it by a polynomial $P \in \P_n$ which has the representation
$$ P(x) = \sum_{k=0}^n y_k \, p_{n,k}(x) $$
with respect to $\B_n$. The polynomial $P$ could be the Bernstein polynomial $B_n(f)$ of $f$ when it is defined, but it can be a replacement, such as the Kantorovich polynomial of $f$ (\cite[Ch.10]{DL}), or it can also be a better approximant in $\P_n$, especially if $f$ has high order of smoothness. Both the quality of the approximation and the range of the coefficients $(y_k)_0^n$ will matter.
We define the {\em approximation error} of $f$ by
$$ E^A_n(f,P;x) := f(x) - P(x), $$
and will often suppress $P$ in the notation. 

The second stage consists of quantization of $P$ via its representation in the Bernstein basis. The coefficients $y:=(y_k)_0^n$ of $P$ will be replaced with their quantized version $q:=(q_k)_0^n$ taking values in a constrained set $\A$, called the (quantization) alphabet. We define the {\em quantization error} of $y$ by
$$ E^Q_n(y,q;x):=\sum_{k=0}^n (y_k - q_k)p_{n,k}(x), $$
and again, will often suppress $q$ in our notation.

\subsection{Noise-shaping through $\Sigma\Delta$ quantization:}

$\Sigma\Delta$ quantization (modulation) refers to a large family of algorithms designed to convert any given sequence $y:=(y_k)$ of real numbers to another sequence $q:=(q_k)$ taking values in a discrete set $\A$ (typically an arithmetic progression) in such a way that the error $y-q$ between them (i.e. the quantization error) is a ``high-pass'' sequence, i.e. it produces a small inner-product with any slowly varying sequence. The canonical way of ensuring this is to ask that $y$ and $q$ satisfy the $r$th order difference equation
\begin{equation} \label{sigmadeltaeq}
y - q = \Delta^r u, 
\end{equation}
where $(\Delta u)_k:= u_k - u_{k-1}$, for some bounded sequence $u$. We then say that $q$ is an $r$th order noise-shaped quantization of $y$. When \eqref{sigmadeltaeq} is implemented recursively, it means that each $q_k$ is found by means of a ``quantization rule'' of the form
\begin{equation} \label{quantrule}
 q_k = F(u_{k-1},u_{k-2},\dots,y_k,y_{k-1},\dots),
\end{equation}
and $u_k$ is updated via
\begin{equation} \label{sigmadeltaupdate}
 u_k = \sum_{j=1}^r (-1)^{j-1} \binom{r}{j} u_{k-j} + y_k - q_k
\end{equation}
to satisfy \eqref{sigmadeltaeq}.
This recursive process is commonly called ``feedback quantization'' due to the role $q_k$ plays as a feedback control term in \eqref{sigmadeltaupdate}. The role of the quantization rule \eqref{quantrule} is to keep the solution $u$ bounded for all input sequences $y$ of arbitrary duration in a given set $\Y$. In this case, we say that the quantization rule is stable for $\Y$.

The ``greedy'' quantization rule refers to the function $F$ which outputs any minimizer $q_k \in \A$ of $|u_k|$ as determined by \eqref{sigmadeltaupdate}. More precisely, it sets 
\begin{equation}
 q_k:=\mathrm{round}_\A \left (\sum_{j=1}^r (-1)^{j-1} \binom{r}{j} u_{k-j} + y_k\right)
\end{equation}
where $\mathrm{round}_\A(v)$ stands for any element in $\A$ that is closest to $v$.
If $\A$ is an infinite arithmetic progression of step size $\delta$, then greedy quantization is always stable (i.e. for any $r$ and for all input sequences) and produces a solution $u$ to \eqref{sigmadeltaeq} via \eqref{sigmadeltaupdate} which is bounded by $\delta/2$. 

Real difficulties start when $\A$ is a fixed, finite set, the extreme case being a set of two elements, which can be taken to be $\{\pm 1\}$ without loss of generality. In this case, it is a major challenge to design quantization rules that are stable for any order $r$ and for arbitrary bounded inputs in a range $[-\mu,\mu]$. The first breakthrough on this problem was made in the seminal paper of Daubechies and DeVore \cite{DD} where it was shown for $\A = \{\pm 1\}$ that for any order $r\geq 1$ and for any $\mu < 1$, there is a stable $r$th order quantization rule $F_{r,\mu}$ which guarantees that $u$ is bounded by some constant $C_{r,\mu}$ for all input sequences $y$ that are bounded by $\mu$. (For $r=1$, one can take $\mu = 1$ and find that $C_{1,1} = 1$ will do.) The constant $C_{r,\mu}$ depends on both $r$ and $\mu$, and blows up as $r \to \infty$, or as $\mu \to 1^-$ (except when $r=1$). Another family of stable quantization rules, but with more favorable $C_{r,\mu}$, was subsequently proposed in \cite{exp_decay}. In this paper we will be using the greedy quantization rule when $\A = \Z$, and either of the rules in \cite{DD} and \cite{exp_decay} when $\A = \{\pm 1\}$. We will not need the explicit descriptions of these rules, or of the associated bounds $C_{r,\mu}$. 

Clearly $q$ is a bounded sequence when $\A$ is a finite set, but even when $\A = \Z$, $q$ will be bounded when the input $y$ is bounded. This is because \eqref{sigmadeltaeq} implies $\|y - q \|_\infty \leq 2^r \|u\|_\infty \leq 2^{r-1}$ so that $\|q\|_\infty \leq \|y\|_\infty + 2^{r-1}$.

In the remaining sections, the letter $C$ will represent any constant whose value is not important for the discussion. Its value may be updated and it may also stand for different constants. Whenever $C$ depends on a given set of parameters, as in the previous paragraph, these parameters will be specified.

\subsection{Effect of noise shaping in the Bernstein basis}

It will be convenient for us to extend the index $k$ beyond $n$ in the definition of $p_{n,k}$. We do this without modifying the formula \eqref{Bernsteinbasis}, noting that this implies $p_{n,k} = 0$ for all $k > n$ since $\binom{n}{k} = 0$ in this range. 

As indicated in the previous subsection, we will always work with a 
stable $r$th order $\Sigma\Delta$ quantization scheme that is applied to convert an input sequence $y:=(y_k)_0^n$ bounded by $\mu$ to its quantized version $q:=(q_k)_0^n$. We will set $u_k = 0$ for all $k < 0$. With this assumption, we have the total quantization error
\begin{eqnarray}
 E^Q_n(y;x) & = & \sum_{k = 0}^n (\Delta^r u)_k \,p_{n,k}(x) \nonumber \\
 & = & \sum_{k=0}^n u_k \,(\tilde \Delta^r p_{n,\cdot}(x))_k \label{summedbyparts}
\end{eqnarray}
where $\tilde \Delta$ is the adjoint of $\Delta$ given by $(\tilde \Delta u)_k := u_k - u_{k+1}$. Notice that with our convention on $p_{n,k}$ for $k > n$, all the boundary terms are correctly included in \eqref{summedbyparts}. It now follows by our assumption of stability that $|u_k| \leq C_{r,\mu}$ for all $k$, so that 
\begin{equation}\label{genericbound}
 |E^Q_n(y;x) | \leq C_{r,\mu} V_{n,r}(x)
\end{equation}
where 
\begin{equation} \label{Bernsteinvariation}
V_{n,r}(x) := \sum_{k=0}^n \left |(\tilde \Delta^r p_{n,\cdot}(x))_k\right|
\end{equation}
stands for the {\em $r$th order variation} of the basis $\B_n$.

As we consider increasing values of $r$ in the sections below, we will be providing specific upper bounds for $V_{n,r}(x)$ of increasing complexity. We now note two important properties of the Bernstein basis that we will employ in our analysis. (See \cite[Ch.10]{DL} and \cite{lorentz1} for these and other useful facts about Bernstein polynomials.)
\begin{enumerate}
 \item The consecutive differences of the $p_{n,k}$ satisfy
\begin{equation}\label{Bernsteindiff}
 (\tilde \Delta p_{n,\cdot}(x))_k = p_{n,k}(x) - p_{n,k+1}(x) = \frac{(k+1)-(n+1)x}{(n+1)x(1-x)}
 p_{n+1,k+1}(x)
\end{equation}
which, with our convention on $p_{n,k}$ for $k>n$, holds for all $n,k \geq 0$, and all $x\in[0,1]$. When interpreting the right hand side of this equation for $x=0$ or $x=1$, it should be observed that the polynomial $ ((k{+}1){-}(n{+}1)x)p_{n+1,k+1}(x)$ is divisible by $x(1-x)$ for all $k \geq 0$. As is customary, we will use the short notation $X:=x(1-x)$.

\item We set
\begin{equation}
 T_{n,s}(x) := \sum_{k=0}^n (k-nx)^{s} p_{n,k}(x)
\end{equation}
for each non-negative integer $s$. Then we have $T_{n,0}(x) = 1$, $T_{n,1}(x) = 0$, $T_{n,2}(x) = nX$. In general, for each $s$, there is a constant $A_s$ such that 
\begin{equation} \label{Tbound}
0\leq  T_{n,2s}(x) \leq A_s n^s
\end{equation}
holds for all $n \geq 1$, uniformly over $x \in [0,1]$. 
\end{enumerate}

\section{Specific error bounds}\label{specificbounds}
\subsection{First order noise shaping}

For first order $\Sigma\Delta$ quantization, the greedy rule is the only rule one needs. 
There will be two choices for $\A$ of interest to us, $\{\pm 1\}$ and $\Z$. When $\A = \{\pm 1\}$, we have $|u_k| \leq 1$ for all input sequences $y$ bounded by $1$. When $\A = \Z$, we have $|u_k| \leq \frac{1}{2}$ for all input sequences $y$. In either case, \eqref{genericbound}, \eqref{Bernsteinvariation} and 
\eqref{Bernsteindiff} give us, for all admissible input sequences $y$, the bound
\begin{eqnarray}
|E^Q_n(y;x)| & \leq & V_{n,1}(x) \nonumber \\
& = &
\frac{1}{(n+1)X} \sum_{k=0}^n \big |(k{+}1)-(n{+}1)x \big |\,p_{n+1,k+1}(x) \nonumber \\
& \leq & \frac{1}{(n+1)X} \sqrt{T_{n+1,2}(x)} \nonumber \\
& = &  \frac{1}{\sqrt{(n+1)X}}, \label{quantbound0}
\end{eqnarray}
where we used Cauchy-Schwarz inequality in the third line.
We also have the trivial bound
\begin{equation} \label{trivialbound}
|E^Q_n(y;x)| \leq  \|y - q \|_\infty \leq 2\|u\|_\infty \lesssim 1.
\end{equation}
(Here, as usual, $A_n \lesssim B_n$ means $A_n \leq C B_n$ for all $n$ where $C$ is an absolute constant. When $C$ depends on some parameter $\alpha$, we use the notation $\lesssim_\alpha$.) Combining these two bounds, we obtain
\begin{equation}\label{quantbound1}
|E^Q_n(y;x)| \lesssim  \min\left ( \frac{1}{\sqrt{(n+1)X}}, 1 \right ) \lesssim \frac{1}{1+\sqrt{nX}}.
\end{equation}
Our first theorem follows directly from this bound:
\begin{theorem}\label{theorem1st}
For every continuous function $f:[0,1] \to [-1,1]$, and for every positive integer $n$, there exist signs $\sigma_0,\dots,\sigma_n \in \{\pm 1\}$ such that 
\begin{equation}
\left | f(x) - \sum_{k=0}^n \sigma_k \,p_{n,k}(x) \right | \lesssim
\omega_2(f,\sqrt{X/n}) + \frac{1}{1 + \sqrt{nX}},
\end{equation}
where $\omega_2$ stands for the $2$nd modulus of smoothness of $f$ and $X:=x(1-x)$. In particular, if $f$ is Lipschitz, then
\begin{equation} \label{boundLip1st}
\left | f(x) - \sum_{k=0}^n \sigma_k \,p_{n,k}(x) \right | \lesssim
\frac{1+|f|_\mathrm{Lip}}{1 + \sqrt{nX}}.
\end{equation}
\end{theorem}

\begin{proof}
For any given $n$, we set $y_k = f(\frac{k}{n})$, $k=0,\dots,n$, so that $P(x) = B_n(f,x)$, and we use $\sigma_k = q_k$ as the output of the first order $\Sigma\Delta$ scheme with input $(y_k)$. Since $|y_k| \leq 1$ for all $k$, the scheme is stable and the result follows by noting (see \cite[p.308]{DL} that 
$$|E^A(f;x)| = |f(x)  - B_n(f,x)| \lesssim \omega_2(f,\sqrt{X/n})$$
together with \eqref{quantbound1}. The specific bound for the Lipschitz class follows trivially by noting that $\sqrt{X/n} \leq 1/(1+\sqrt{nX})$.
\end{proof}

{\bf Remarks:} 
\begin{itemize}
 
 \item As a function of $x$, the upper bound \eqref{boundLip1st} goes to $0$ at the rate $n^{-1/2}$ in $C([a,b])$, $0<a<b<1$, and in $L^p([0,1])$, $1\leq p<2$. The case $p\geq 2$ yields slower decay rates. However, a complete discussion of functions in $L^p([0,1])$ requires an unbounded alphabet, since the approximating polynomials (whether one uses suitable substitute Bernstein polynomials, or the Kantorovich polynomials of $f$ directly), can have arbitrarily large coefficients in the Bernstein basis. See Theorem \ref{theorem1st2} below for the full discussion of this case for $\A = \Z$, which includes a more precise description of the quantization error in all $p$-norms.
 
 \item The convergence rate $n^{-1/2}$ is the best one can expect from Bernstein polynomial approximation of Lipschitz functions (even without quantization), so it may come as a surprise that the rate holds up in the case of one-bit quantization. Information-theoretically, however, this rate is suboptimal, though perhaps not hugely so: The $\varepsilon$-entropy of the set of functions 
 $$\big \{f \in \mathrm{Lip}([a,b]) : \|f\|_\infty \leq 1, ~ |f|_\mathrm{Lip} \leq L \big\}$$  
 behaves as $L(b-a)\varepsilon^{-1}(1+ o(1))$ (see \cite{KT, lorentz2}) which implies that any $\varepsilon$-net for it of cardinality $O(2^n)$ must satisfy $\varepsilon \gtrsim L(b-a)n^{-1}$. 
 
 \item We used the Bernstein polynomial of $f$, i.e. $y_k = f(k/n)$, in our construction to achieve two goals at once: (i) approximation and (ii) guaranteed stability of the $\Sigma\Delta$ quantization scheme (due to $y_k$ being in the same range as $f$). However, if one can find better polynomial approximants whose coefficients in the Bernstein basis still fall in the range $[-1,1]$, then these could be used instead to reduce the approximation error $E^A(f,P;x)$. We will discuss this in Section \ref{beyond}. (We also note here a relevant result: If a polynomial $P(x)$ takes values in $(-1,1)$ for all $x \in [0,1]$, then \cite[Thm 1]{Rosenberg} shows that for all sufficiently large $m$, the coefficients of $P$ with respect to $\B_m$ will also be in the range $(-1,1)$. This opens up the possibility of employing better $P$, provided the value of $m$ can be controlled well, relative to the degree of $P$. However, it seems that in this generality the currently available quantitative information regarding this question is not useful enough, at least not immediately, to yield an improvement.)
 
 \item As we will see in the next two subsections, the quantization error $E^Q(y,q;x)$ can be reduced via higher order schemes.
 \end{itemize}

The difficulty about improving the approximation error $E^A(f,P;x)$ while maintaining quantizer stability disappears when $\A = \Z$ since the greedy $\Sigma\Delta$ scheme is then unconditionally stable (i.e. for all inputs), and this leads to a stronger approximation result as we show next.
For any $p \in [1,\infty)$, let $E_n(f)_p$ stand for the error of best approximation of $f \in L^p([0,1])$ by elements of $\P_n$. For $p=\infty$, we consider $f \in C([0,1])$ instead. We have

\begin{theorem}\label{theorem1st2}
For every real-valued function $f \in L^p([0,1])$, where $1\leq p < \infty$, and for every positive integer $n$, there exist $q_0,\dots,q_n \in \Z$ such that 
\begin{equation} \label{Lpbound1st}
\left \| f - \sum_{k=0}^n q_k \,p_{n,k} \right \|_p \lesssim
E_n(f)_p + C_p \left\{ \begin{array}{ll}
                    n^{-1/2}, & 1 \leq p < 2, \\
                    \sqrt{\frac{\log n}{n}}, & p = 2, \\
                    n^{-1/p}, & 2 < p < \infty.
                   \end{array}
 \right.
\end{equation}
If $f$ is continuous, then there exist $q_0,\dots,q_n \in \Z$ such that
\begin{equation}
\left | f(x) - \sum_{k=0}^n q_k \,p_{n,k}(x) \right | \lesssim
E_n(f)_\infty + \frac{1}{1 + \sqrt{nX}}.
\end{equation}
Consequently, for any sequence $(n_j)_1^\infty$ of integers such that $n_j \to \infty$, $\displaystyle \bigcup_{j=1}^\infty \L(\B_{n_j})$ is dense in $L^p([0,1])$ for all $p \in [1,\infty)$, and in
$C([a, b])$ for all $0<a<b<1$.
\end{theorem}
\begin{proof}
Given $f$ in $L^p([0,1])$ or $C([0,1])$, let $P\in\P_n$ be the best approximant of $f$ in the corresponding $p$-norm, with coefficients $y:=(y_k)_0^n$ in the Bernstein basis. Hence we have $\|E^A_n(f)\|_p = E_n(f)_p$. 

Next we bound $\|E^Q_n(y)\|_p$ using \eqref{quantbound1}. Because of the symmetry of the bound with respect to $1/2$, and employing the inequality $x(1-x) \geq x/2$ for $x\in [0,1/2]$, we have
\begin{eqnarray}
\int_0^1 \frac{1}{(1+\sqrt{nX})^p} \,\mathrm{d}x 
& \leq & 2 \int_0^{1/n} \mathrm{d}x +
2 \int_{1/n}^{1/2} \frac{1}{(\sqrt{nx/2})^p}\,\mathrm{d}x \nonumber \\
& \lesssim_p &
 \left\{ \begin{array}{ll}
                    n^{-p/2}, & 1 \leq p < 2, \\
                    \frac{\log n}{n}, & p = 2, \\
                    n^{-1}, & 2 < p < \infty,
                   \end{array}
 \right. 
\end{eqnarray}
hence \eqref{Lpbound1st} follows.
When $p = \infty$, we have to settle with the pointwise bound only. 
\end{proof}

\par \noindent {\bf Remark:}
With additional information on $f$, the approximation error term $E_n(f)_p$ in Theorem \ref{theorem1st2} can be quantified further; for example, via $E_n(f)_p \lesssim_s \omega_s(f,1/n)_p$, which holds for any $n,s$ such that $n \geq s$.

\subsection{Second order noise shaping}
For the second order, we will use the greedy quantization rule for $\A = \Z$, and any of the stable schemes proposed in \cite{DD} or \cite{exp_decay}. The analysis of the quantization error will be the same in both cases, the only difference concerning the range of admissible coefficients $y$. For $\A = \Z$, we will be able to work with all sequences again, whereas for $\A = \{ \pm 1\}$ we will have to restrict to the $y_k$ that range in $[-\mu,\mu]$ for some $\mu < 1$. The trivial bound \eqref{trivialbound} is replaced with
\begin{equation} \label{trivialbound2}
 |E^Q(y;x)| \leq \|y - q \|_\infty = \|\Delta^2 u \|_\infty \leq 4 \|u\|_\infty \leq
 \left\{ \begin{array}{ll}
                    4, & \mbox{ if } \A = \Z, \\
                    4 C_{2,\mu}, & \mbox{ if } \A = \{\pm 1\}, 
                   \end{array}
 \right. 
\end{equation}
where the value of $C_{2,\mu}$ depends on which stable second order scheme is employed.

Meanwhile, utilization of the general purpose bound \eqref{genericbound} via \eqref{Bernsteinvariation} requires another application of $\tilde \Delta$ on \eqref{Bernsteindiff}. To lighten the notation, we will adopt the following convention when there is no possibility for confusion: $p_{n,k}$ will be short for $p_{n,k}(x)$ (since all basis functions will be evaluated at the same point), and $\tilde \Delta^r p_{n,k}$ will be short for $(\tilde \Delta^r p_{n,\cdot}(x))_k$. (Note that there is no ambiguity in the meaning of $\tilde \Delta$ acting on a shifted sequence such as $p_{n,k+1}$ since shifts commute with $\tilde \Delta$.) With this convention, we have
\begin{eqnarray}
 \tilde \Delta^2 p_{n,k} & = & \tilde \Delta p_{n,k} - \tilde \Delta p_{n,k+1}
 \nonumber \\
 & = & \frac{1}{(n{+}1)X} \Big( \big (k{+}1 - (n{+}1)x\big ) p_{n+1,k+1} - \big (k{+}2 - (n{+}1)x\big )p_{n+1,k+2} \Big) \nonumber \\
 & = & \frac{1}{(n{+}1)X} \Big( \big (k{+}2 - (n{+}1)x\big ) \tilde \Delta p_{n+1,k+1} - p_{n+1,k+1} \Big)
 \nonumber \\
 & = & \frac{1}{(n{+}1)X} \Big( \big (k{+}2 - (n{+}2)x\big ) \tilde \Delta p_{n+1,k+1} + x \tilde \Delta p_{n+1,k+1} - p_{n+1,k+1} \Big)
 \nonumber \\
  & = & \frac{1}{(n{+}1)(n{+}2)X^2} \big (k{+}2 - (n{+}2)x\big )^2 p_{n+2,k+2}
  -\frac{1}{(n{+}1)X} \Big( x p_{n+1,k+2} +(1-x) p_{n+1,k+1} \Big) \nonumber
\end{eqnarray}
so that 
\begin{equation}
V_{n,2}(x)=\sum_{k=0}^n  \left |  (\tilde \Delta^2 p_{n,\cdot}(x))_k \right | \leq
\frac{T_{n+2,2}(x)}{(n{+}1)(n{+}2)X^2} + \frac{1}{(n+1)X} = \frac{2}{(n+1)X}
\end{equation}
and therefore 
\begin{equation}\label{quantbound2}
|E^Q_n(y;x)| \lesssim_\mu  \min\left ( \frac{1}{(n+1)X}, 1 \right ) \lesssim \frac{1}{1+nX}.
\end{equation}

Equipped with this bound, we can now improve the rate of convergence for smoother functions. The proof of the next theorem mirrors the proof of Theorem \ref{theorem1st} verbatim.
\begin{theorem}\label{theorem2nd}
Let $\mu < 1$ be arbitrary. For every continuous function $f:[0,1] \to [-\mu,\mu]$, and for every positive integer $n$, there exist signs $\sigma_0,\dots,\sigma_n \in \{\pm 1\}$ such that 
\begin{equation}
\left | f(x) - \sum_{k=0}^n \sigma_k \,p_{n,k}(x) \right | \lesssim_\mu
\omega_2(f,\sqrt{X/n}) + \frac{1}{1 + nX}
\end{equation}
where $\omega_2$ stands for the $2$nd modulus of smoothness of $f$ and $X:=x(1-x)$. In particular, if $f$ has a Lipschitz derivative, then
\begin{equation}
\left | f(x) - \sum_{k=0}^n \sigma_k \,p_{n,k}(x) \right | \lesssim_\mu
\frac{1+|f'|_\mathrm{Lip}}{1 + nX}.
\end{equation}
\end{theorem}

Similarly, we are now also able to improve the bound achieved in Theorem \ref{theorem1st2}. Again, the proof of next theorem mirrors that of Theorem \ref{theorem1st2}. The only difference is now we are bounding the $p$-norm of the function $1/(1+nX)$.
\begin{theorem}\label{theorem2nd2}
For every real-valued function $f \in L^p([0,1])$, where $1\leq p < \infty$, and for every positive integer $n$, there exist $q_0,\dots,q_n \in \Z$ such that 
\begin{equation} \label{Lpbound2nd}
\left \| f - \sum_{k=0}^n q_k \,p_{n,k} \right \|_p \lesssim
E_n(f)_p + C_p \left\{ \begin{array}{ll}
                    \frac{\log n}{n}, & p = 1, \\
                    n^{-1/p}, & 1 < p < \infty. 
                   \end{array}
 \right.
\end{equation}
If $f$ is continuous, then
\begin{equation}
\left | f(x) - \sum_{k=0}^n q_k \,p_{n,k}(x) \right | \lesssim
E_n(f)_\infty + \frac{1}{1 + nX}.
\end{equation}
\end{theorem}
\subsection{Going beyond second order} \label{beyond}

Inspecting the outcome of the first and second order noise shaping, it is natural to expect faster decay of the quantization error when higher order noise-shaping is employed. This expectation can be met, as we will discuss later in this section. But first, we need to address the approximation error. As long as $B_n(f)$ is employed in the first approximation stage, the error bound of Theorem \ref{theorem2nd} is the best one can achieve (i.e. regardless of the amount of smoothness of $f$) due to the saturation of the approximation error at the rate $1/n$; any gain in the quantization error from using higher order noise shaping would be drowned by this term. Possibility of improvement by means of better polynomial approximations whose coefficients in $\B_n$ also admit stable higher order quantization remains. Indeed, as we did before, this approach works ``out of the box'' in the case $\A = \Z$ since we have unconditional stability. However, for $\A=\{\pm 1\}$, this is not immediate: if we are to work with our stability criterion $\|y\|_\infty \leq \mu < 1$, then we will need to ensure that the coefficients of these improved polynomial approximations also satisfy this criterion. 

Several adjustments of classical Bernstein polynomials that adapt to the smoothness of the function that is being approximated have been proposed in the literature. Most of these are not suitable for our quantization method, at least not immediately, as they do not necessarily produce polynomial approximations of the same degree, but there is one that works: the particular linear combinations of iterated Bernstein operators, as developed by Micchelli \cite{micchelli1973} and Felbecker \cite{felbecker1979}. As before, let $B_n$ be the Bernstein operator on $C([0,1])$ that maps $f$ to $B_n(f)$, and for any integer $r \geq 1$, let
\begin{equation}
U_{n,r} := I - (I-B_n)^r. 
\end{equation}
It follows by combining the results of \cite{micchelli1973} and \cite{felbecker1979} (see also \cite[Ch.9.2]{bustamante}) that for any $s \geq 1$, if $f \in C^{s-1}([0,1])$ and $f^{(s-1)}$ is Lipschitz, then 
\begin{equation}\label{U-bound}
\|f - U_{n,\lceil s/2 \rceil}(f) \|_\infty \lesssim_s \|f\|_{_{C^{s-1}\mathrm{Lip}}} n^{-s/2}
\end{equation}
where $\lceil \cdot \rceil$ denotes the ceiling operator and 
\begin{equation}\label{CsLipnorm}
\|f\|_{_{C^{s-1}\mathrm{Lip}}}:=\max(\|f\|_{C^{s-1}}, |f^{(s-1)}|_\mathrm{Lip}).
\end{equation}
Note that $U_{n,1} = B_n$, so for $s=1$ and $s=2$, the result \eqref{U-bound} reduces to the regular Bernstein polynomial approximation error bound that we have already utilized.

More generally, for any $r \geq 1$, the coefficients of $U_{n,r}(f)$ with respect to $\B_n$ are well controlled. We have the following:
\begin{theorem}\label{U-stability}
For any $f \in C([0,1])$ and integers $n,r \geq 1$, there exists $f_{n,r}\in C([0,1])$ satisfying
$$U_{n,r}(f) = B_n(f_{n,r}) = \sum_{k=0}^n f_{n,r}\Big( \frac{k}{n} \Big) p_{n,k} $$
and
$$\|f_{n,r} - f \|_\infty \leq (2^{r-1}{-}1) \|f - B_n(f)\|_\infty.$$
Therefore, $f_{n,r} \to f$ uniformly as $n \to \infty$, and in particular, 
$\|f\|_\infty < 1$ implies $\|f_{n,r}\|_\infty < 1$ for all sufficiently large $n$. Furthermore, there is an absolute constant $c_0$ such that if $f\in C^2([0,1])$ and $\|f\|_\infty \leq 1-2\epsilon < 1$, then for all $n \geq \epsilon^{-1} c_0 (2^{r-1}{-}1) \|f^{(2)}\|_\infty $, we have $ \|f_{n,r}\|_\infty \leq 1 - \epsilon$.
\end{theorem}
\begin{proof}
The case $r=1$ is trivial, since then $U_{n,1} = B_n$ and we can set $f_{n,1}=f$. We consider $r \geq 2$.
Note the polynomial identity $1-x^r = (1-x)(1+x+\dots+x^{r-1})$ which implies the relation
$$ U_{n,r} = I-(I-B_n)^r = B_n\Big(I + \sum_{j=0}^{r-2} (I-B_n)^{j+1} \Big).$$
Hence, defining $$f_{n,r} := f + \sum_{j=0}^{r-2} (I-B_n)^{j+1}(f)$$ it follows at once that
$$ U_{n,r}(f) = B_n(f_{n,r})$$
and that
$$ \|f_{n,r}-f\|_\infty \leq 
\sum_{j=0}^{r-2} \|(I - B_n)^j\|_{\infty \to \infty} \|f - B_n(f)\|_\infty. $$
At this point we only need to utilize the trivial bound $\|(I-B_n)^j\|_{\infty \to \infty} \leq 2^j$ which follows from $\|B_n\|_{\infty \to \infty} = 1$. Hence we get
$$ \|f_{n,r}-f\|_\infty \leq 
(2^{r-1}{-}1)\|f - B_n(f)\|_\infty.$$
The next statement is now immediate. In particular, if $f \in C^2([0,1])$, we can use the bound $\|f - B_n(f)\|_\infty 
\leq c_0 \|f^{(2)}\|_\infty n^{-1} $ to reach the final conclusion.
\end{proof}

With the above theorem, we are able to enjoy approximation error $E^A(f,P)=O(n^{-s/2})$ for $f \in C^{s-1}\mathrm{Lip}([0,1])$ while maintaining stability of the quantization scheme (for all sufficiently large $n$) for the case $\A = \{\pm 1\}$ with the mere condition $\|f \|_\infty < 1$. 
We can now return to reduction of quantization error by means of higher order noise shaping. Let us first check the ansatz 
$$\beta_{n,r}(x) := C_{r,\mu} \min\left(\frac{1}{((n+1)X)^{r/2}},1\right)$$ 
as a bound for the $r$th order quantization error. It is apparent that as a function of $n$,
$$ \|\beta_{n,r}\|^p_p \asymp \| \beta_{n,2}\|^{rp/2}_{rp/2} \asymp \frac{1}{n} $$
for all $r > 2$ and all $1 \leq p < \infty$. Therefore, this type of bound would offer no gain over the second order case in terms of the rate of convergence in $p$-norms, except for the removal of the $\log n$ factor for $p=1$. 
However, the situation is not as disappointing for the pointwise bounds (or the uniform bounds on subintervals). Our goal is now to provide a path towards validating this ansatz.

We will again use our notation convention and write $\tilde \Delta^r p_{n,k}$ for $(\tilde \Delta^r p_{n,\cdot}(x))_k$.
For $r \geq 1$, we have 
\begin{eqnarray}
 \tilde \Delta^{r+1} p_{n,k}
 & = & \frac{1}{(n{+}1)X}\,\tilde \Delta^r  \Big( (k{+}1 - (n{+}1)x)\,p_{n+1,k+1} \Big) \nonumber \\
 & = & \frac{1}{(n{+}1)X}\,\Big( (k{+}1 - (n{+}1)x)\, \tilde \Delta^r p_{n+1,k+1} - r \tilde \Delta^{r-1} p_{n+1,k+2} \Big) \label{3termrec}
\end{eqnarray}
where in the second equality we have used the Leibniz formula
$$ \tilde \Delta^r (a_k b_k) = \sum_{j=0}^r \binom{r}{j} (\tilde \Delta^j a_k)( \tilde \Delta^{r-j} b_{k+j})$$
with $a_k =  (k{+}1 - (n{+}1)x)$ and $b_k = p_{n+1,k+1}$. Notice that $\tilde \Delta^j a_k = 0$ for $j \geq 2$.

The recurrence relation \eqref{3termrec} paves the way for induction for bounding $V_{n,r}(x)$. However, due to the presence of $(k{+}1 - (n{+}1)x)$, we will work with a stronger induction hypothesis concerning all of the sums
$$ Y_{n,r,s}(x) := \sum_{k\geq0} |k-nx|^s \left | \tilde \Delta^r p_{n,k}(x) \right |  $$
at once. Note that this sum runs over all $k \geq 0$ but there is no contribution from the terms $k > n$.

\begin{theorem}
For all non-negative integers $r$ and $s$, we have
\begin{equation}\label{variationmoments}
Y_{n,r,s}(x) \lesssim_{r,s} n^{\frac{s-r}{2}} X^{-r}.
\end{equation}
In particular, 
\begin{equation}\label{Vnrbound}
V_{n,r}(x) = Y_{n,r,0}(x) \lesssim_r n^{-\frac{r}{2}} X^{-r}.
\end{equation}
\end{theorem}
\par \noindent {\bf Remark.} Note that in its $X$ dependence, this is a weaker bound than our ansatz. However, the $n$ dependence is the same. The exponent of $X$ in the proof below can be improved, but we do not know if $-r/2$ is achievable for $r > 2$. However, given that we are primarily interested in the rate of convergence relative to $n$, this question is of secondary importance and we leave it for future work.
\begin{proof}
For any $r$, let $\mathbf{P}(r)$ be the statement that \eqref{variationmoments} holds for all $s$ (as well as $n$ and $x$). The case $r=0$ is readily covered by \eqref{Tbound} and Cauchy-Schwarz:
$$ Y_{n,0,s}(x) \leq \sqrt{T_{n,2s}(x)} \lesssim_s n^{\frac{s}{2}}.$$
The case $r=1$ is handled by noticing that 
\begin{eqnarray}
Y_{n,1,s}(x) &=&  \frac{1}{(n{+}1)X} \sum_{k=0}^n |k-nx|^s \,|k{+}1-(n{+}1)x| \,p_{n+1,k+1}(x) \nonumber \\
& \leq &\frac{1}{(n{+}1)X} \sum_{k=0}^n 2^s \big(|k{+}1-(n{+}1)x|^s + |1-x|^s\big) \,|k{+}1-(n{+}1)x| \,p_{n+1,k+1}(x) \nonumber \\
& \lesssim_s &\frac{1}{(n{+}1)X} \left ( Y_{n+1,0,s+1}(x) + Y_{n+1,0,s}(x) \right) \nonumber \\
& \lesssim_s & n^\frac{s-1}{2} X^{-1}.
\end{eqnarray}
Assume that $\mathbf{P}(r)$ and $\mathbf{P}(r{-}1)$ are true. We will verify $\mathbf{P}(r{+}1)$. 
The identity \eqref{3termrec} implies that 
\begin{eqnarray}
|k-nx|^s \left|\tilde \Delta^{r+1} p_{n,k}\right |
 & \lesssim_s & (|k{+}1-(n{+}1)x|^s+1) \frac{1}{(n{+}1)X}\,|k{+}1 - (n{+}1)x|\, \left|\tilde \Delta^r p_{n+1,k+1}\right|
 \nonumber \\
 & & +(|k{+}2-(n{+}1)x|^s+2^s) \frac{r}{(n{+}1)X}\,\left|\tilde \Delta^{r-1} p_{n+1,k+2} \right|  \nonumber  
\end{eqnarray}
so that summing over all $k \geq 0$, and applying our induction hypothesis, we get
\begin{eqnarray}
Y_{n,r+1,s}(x) & \lesssim_{r,s} & \frac{1}{(n{+}1)X} \left ( Y_{n+1,r,s+1}(x) + Y_{n+1,r,1}(x) 
+ Y_{n+1,r-1,s}(x) + Y_{n+1,r-1,0}(x) \right) \nonumber \\
& \lesssim_{r,s} & \frac{1}{(n{+}1)X} \left ( n^\frac{s+1-r}{2}X^{-r} + 
n^\frac{1-r}{2}X^{-r} + n^\frac{s-r+1}{2}X^{-r+1} + n^\frac{-r+1}{2}X^{-r+1} \right)  \nonumber \\
& \lesssim_{r,s} & n^\frac{s-(r+1)}{2}X^{-(r+1)},
\end{eqnarray}
hence $\mathbf{P}(r{+}1)$ is true and the induction step is complete.
\end{proof}
With the stability guarantee of Theorem \ref{U-stability} and the bounds \eqref{U-bound} and \eqref{Vnrbound} in place for the approximation and the quantization errors, all is left now is to state our final two theorems. The first theorem below follows by employing $P = U_{n,\lceil s/2 \rceil}(f)$, and the second one by employing the best approximation to $f$ from $\P_n$.
\begin{theorem}\label{theorem3rd}
Let  $s \geq 3$ be any integer. For every $f \in C^{s-1}([0,1])$ with $f^{(s-1)} \in \mathrm{Lip}([0,1])$ and $\|f\|_\infty \leq \mu < 1$,  there exist signs $\sigma_0,\dots,\sigma_n \in \{\pm 1\}$ such that 
\begin{equation}
\left | f(x) - \sum_{k=0}^n \sigma_k \,p_{n,k}(x) \right | \lesssim_{\mu,s}
\|f\|_{C^{s-1}\mathrm{Lip}}n^{-s/2} + \min(1,X^{-s} n^{-s/2})
\end{equation}
holds for all $n \gtrsim\|f^{(2)}\|_\infty / (1-\mu)$ where $\|f\|_{C^{s-1}\mathrm{Lip}}$ is defined in \eqref{CsLipnorm} and $X:=x(1-x)$.
\end{theorem}

\begin{theorem} \label{theorem3rd2}
For every continuous function $f:[0,1] \to \R$ and integers $r \geq 0$, $n \geq 1$, there exist $q_0,\dots,q_n \in \Z$ such that
\begin{equation}
\left | f(x) - \sum_{k=0}^n q_k \,p_{n,k}(x) \right | \lesssim_r
E_n(f)_\infty + \min(1,X^{-r}n^{-r/2})
\end{equation}
for all $x \in [0,1]$ where $X:=x(1-x)$. In particular, if $f \in C^s([0,1])$, then for all $n \geq 1$, there exist $q_0,\dots,q_n \in \Z$ such that
\begin{equation}
\left | f(x) - \sum_{k=0}^n q_k \,p_{n,k}(x) \right | \lesssim_s
\|f^{(s)}\|_\infty n^{-s} + \min(1,X^{-2s}n^{-s})
\end{equation}
for all $x \in [0,1]$. 
\end{theorem}

\section*{Appendix A: Frame-theoretic redundancy of the Bernstein basis}
\label{sec:numerical}

We have shown in this paper that a large class of functions (including all continuous functions) can be approximated arbitrarily well using polynomials with very coarsely quantized coefficients in a Bernstein basis. While our methodology can be viewed as a variation on the theme of \cite{BPY1} which addresses quantization of finite frame expansions through $\Sigma\Delta$ quantization, we emphasize that we are working with linearly independent systems. As such, the setting of this paper is novel for this type of quantization. The reason why the method works is that the ``effective span'' of the Bernstein basis $\B_n$ is actually roughly of dimension $\sim \sqrt{n}$. Both approximation and noise-shaping quantization take place relative to this subspace of $\P_n$.

In this appendix, we will make the statements in the preceding paragraph more precise by first providing a complete description of the singular values of the linear operator 
$S_n\colon \R^{n+1}\to \P_n$ defined by  
$$
S_n u 
:=\sum_{k=0}^n u_k p_{n,k}, 
$$
where $\R^{n+1}$ is equipped with the standard Euclidean inner product 
and $\P_n$ with the inner product inherited from $L^2([0,1])$, both denoted by $\langle \cdot, \cdot \rangle$. Evidently, this discussion is concerned primarily with $L^2$ approximation rather than uniform approximation. However, the description of the singular values shed additional light onto the Bernstein basis and may be of independent interest to the frame community. 

The adjoint $S_n^*\colon \P_n \to \R^{n+1}$ is given by
$$
(S_n^* f)_k = \langle p_{n,k},f\rangle :=\int_0^1 p_{n,k}(x)f(x)\ dx, ~~k=0,\dots,n. 
$$
In the language of frame theory, $S_n$ and $S_n^*$ are the {\em synthesis} and the {\em analysis operators} for the system $\B_n:=(p_{n,k})_{k=0}^n$, respectively. Then the {\em frame operator} $S_n^*S_n^{}\colon \R^{n+1}\to\R^{n+1}$ is given by 
$$ (S_n^*S_n^{} u)_k = \sum_{\ell=0}^n \langle p_{n,k}, p_{n,\ell}\rangle u_\ell = (\Gamma_n u)_k$$
where $\Gamma_n$ is the Gram matrix of the system $\B_n$ with entries
$$
(\Gamma_n)_{k,\ell}
=\langle p_{n,k}, p_{n,\ell}\rangle, ~~k,l=0,\dots,n. 
$$ 
Note that $(n+1) \Gamma_n$ is doubly stochastic due to the relations $$\sum_{l=0}^n p_{n,l}(x) = 1 \mbox{ and } \int_0^1 p_{n,k}(x)\,dx = \frac{1}{n+1}.$$

To ease the notation in the discussion below, we assume $n \geq 0$ is fixed and suppress the index $n$ from our notation when expressing certain quantities, keeping in mind that they still depend on $n$. This dependence will naturally become explicit in all formulas we will provide. Henceforth we will refer to the operators $S:=S_n$ and $\Gamma:=\Gamma_n$. 

Since the Bernstein system is linearly independent, $\Gamma$ is positive-definite, and in particular, invertible. Let its eigenvalues in decreasing order be $(\lambda_k)_0^n$, i.e.
$$
\lambda_{0} \geq \lambda_{1} \geq \cdots \geq \lambda_{n} >0.
$$
As we will see, the eigenvectors of $\Gamma$ are precisely the discrete Legendre polynomials on $\{0,\dots,n\}$, which we denote by $\varphi_0,\dots, \varphi_n$. (Here, we identify $\R^{n+1}$ with 
$\R^{\{0,\dots,n\}}=L^2(\{0,\dots,n\})$ equipped with the counting measure. We pair $f \in \R^{\{0,\dots,n\}}$ with $(f(0),\dots,f(n)) \in \R^{n+1}$.)
Thus each $\varphi_k(\ell)$ is a polynomial in $\ell\in \{0,\dots,n\}$, of degree equal to $k$, and we have the orthonormality relations
$$ \langle \varphi_k, \varphi_j \rangle = \delta_{k,j}.$$
We will not need an explicit formula for the $\varphi_k$. In order to uniquely define the Legendre polynomials, it is necessary to remove their sign ambiguity according some convention, but we shall not need this either.  

For any $k \geq 0$, let $(t)_k$ denote the falling factorial defined by $(t)_k:=t(t-1)\cdots (t-k+1)$ and  $(t)_0:=1$, where we restrict $t$ to non-negative integers.

\begin{theorem}
	\label{thm:Gammaeig}
	Let $n\geq 0$ be arbitrary, $\Gamma$ be the Gram matrix of the Bernstein system $\B_n$. Then for all $0 \leq k \leq n$, we have 
	$ \Gamma \varphi_k = \lambda_k \varphi_k$ where $\varphi_k$ is the degree $k$ discrete Legendre polynomial on $\{0,\dots,n\}$ and
	$$ \lambda_k = \frac{(n)_k}{(n+k+1)_{k+1}}.$$	
\end{theorem}

\begin{proof}
	Let $\pi_k(\ell):=\ell^k$ (as before) and $\xi_k(\ell):=(\ell)_k$, $\ell \in \{0,\dots,n\}$, $k \geq 0$. Since each $\xi_k$ is a monic polynomial of degree $k$, it is immediate that $\mbox{span}\{\xi_0,\dots,\xi_m\} = \mbox{span}\{\pi_0,\dots,\pi_m\} =: P_m$. (Note that for $m>n$ we have $P_m = P_n = \R^{\{0,\dots,n\}}$, though we will not be concerned with the case $m > n$.)

	Let $0 \leq m \leq n$ be arbitrary. We claim that $\Gamma(P_m) \subset P_m$. For this, it suffices to show that $\Gamma \xi_m$ is a degree $m$ polynomial. We have
	\begin{equation}
		\label{eq:eighelper1}
		(\Gamma \xi_m)(k)
		=\sum_{\ell=0}^n \Gamma_{k,\ell} \,\xi_m(\ell)
		= \Big\langle p_{n,k}, \sum_{\ell=0}^n (\ell)_m\, p_{n,\ell}\Big\rangle, ~~~ k=0,\dots,n.
	\end{equation}
	We can evaluate the polynomial sum in this inner-product via
	\begin{equation}
		\label{poly-Bern}
	\sum_{\ell=0}^n (\ell)_m\, p_{n,\ell}(x) = \sum_{\ell=m}^n m! \binom{\ell}{m} \, p_{n,\ell}(x) =
	m! \binom{n}{m}	x^m = (n)_m\,x^m,
	\end{equation}
	where the first equality uses the fact that $(\ell)_m=0$ for $\ell<m$, and the second equality is an identity generalizing the partition of unity property of Bernstein polynomials (see, e.g. \cite[p.85]{bustamante}).
	Hence, \eqref{eq:eighelper1} can now be evaluated to give
	\begin{eqnarray}
		(\Gamma \xi_m)(k)
		& = & (n)_m \binom{n}{k} \int_0^1 x^{k+m} (1-x)^{n-k}  dx \nonumber \\
		& = & (n)_m \binom{n}{k} \mathrm{B}(k+m+1,n-k+1) \nonumber \\
		& = &\frac{(n)_m\,(k+m)_m}{(n+m+1)_{m+1}},~~~ k=0,\dots,n,
		\label{eq:eighelper3}
	\end{eqnarray}
	where $\mathrm{B}(\cdot,\cdot)$ stands for the beta function.
	The term $(k+m)_m$ is a polynomial in $k$ of degree $m$, hence the claim follows.
	
	As a particular case, we have $\Gamma \varphi_m \in P_m$, so there are $\alpha_0^{(m)},\dots,\alpha_m^{(m)}$ such that 
	$$	\Gamma \varphi_m = \sum_{k=0}^m \alpha_k^{(m)} \varphi_k. 	$$
	In other words, $\Gamma$ is represented by an upper-triangular matrix in the discrete Legendre basis $(\varphi_0,\dots,\varphi_n)$. This fact, coupled with the orthogonality of the discrete Legendre basis and the symmetry of $\Gamma$, actually implies that $\Gamma$ is diagonal in the discrete Legendre basis. Let us spell out how this general principle works:
	
	For the case $m=0$, we readily have
	$\Gamma \varphi_0 = \alpha_0^{(0)} \varphi_0$. 
	For $1 \leq m \leq n$, the claim that $\varphi_m$ is an eigenvector of $\Gamma$ is equivalent to the claim  $\alpha_0^{(m)}=\cdots=\alpha_{m-1}^{(m)}=0$. Employing orthogonality and self-adjointness, we have
	\begin{equation}
		\label{eq:eighelper2}
		\alpha_k^{(m)} = \frac{\langle \Gamma \varphi_m, \varphi_k \rangle}{\langle \varphi_k, \varphi_k\rangle} = \frac{\langle \varphi_m, \Gamma \varphi_k \rangle}{\langle \varphi_k, \varphi_k\rangle}, ~~~k=0,\dots,m.
	\end{equation}
	Having
	established that $\varphi_0,\dots,\varphi_{m-1}$ are all eigenvectors of $\Gamma$, \eqref{eq:eighelper2} coupled with orthogonality of the discrete Legendre basis implies that
	$ \alpha_k^{(m)} = 0$ for all $k=0,\dots,m-1$, implying that $\varphi_m$ is an eigenvector of $\Gamma$.
	
	We proceed to obtain a formula for the eigenvalues of $\Gamma$. 	For $m=0$, the fact that $(n+1)\Gamma$ is doubly stochastic implies that $\Gamma \varphi_0 = \lambda_0 \varphi_0$ with $\lambda_0 = 1/(n+1)$. Let $1\leq m \leq n$ and $\lambda_m$ be the eigenvalue correpsonding to the eigenvector $\varphi_m$.
	Since $\varphi_m$, $\xi_m$, $\pi_m$ are all degree $m$ polynomials and the latter two are monic, there is a scalar $\beta_m \not=0$ such that 
	$$ \varphi_m = \beta_m \xi_m + \gamma_{m-1} = \beta_m \pi_m + \tilde \gamma_{m-1}$$
	where $\gamma_{m-1}, \tilde \gamma_{m-1} \in  P_{m-1}$. Hence $$\lambda_m \varphi_m - \beta_m \Gamma \xi_m = \Gamma (\varphi_m - \beta_m \xi_m) = \Gamma \gamma_{m-1} \in P_{m-1}$$
	so that	
	$$ \lambda_m \pi_m - \Gamma \xi_m 
	= \frac{1}{\beta_m}\Big(\lambda_m (\beta_m \pi_m - \varphi_m) + (\lambda_m\varphi_m - \beta_m \Gamma \xi_m)  \Big)
	\in P_{m-1}.$$ 
	Meanwhile, \eqref{eq:eighelper3} implies 
	$$\Gamma \xi_m - \frac{(n)_m}{(n+m+1)_{m+1}} \pi_m  \in P_{m-1}. $$ 
	Adding the last two vectors, we get 
	$$(\lambda_m - \frac{(n)_m}{(n+m+1)_{m+1}}) \pi_m  \in P_{m-1} $$ 
	which implies that $\lambda_m = \frac{(n)_m}{(n+m+1)_{m+1}}$.
\end{proof}

\paragraph{Remark.} For any $n$, let $\sigma_0,\dots,\sigma_n$ be the singular values of $S$ defined by
$\sigma_m := \sqrt{\lambda_m}$. Define $\psi_m := S \varphi_m / \sigma_m$. Then, as is well known, $\psi_0,\dots,\psi_n$ is an orthonormal basis for the range of $S$, i.e. $\P_n$, equipped with the $L^2$ inner product on $[0,1]$. The relation \eqref{poly-Bern} shows that 
$S\xi_m  \in \P_m$, implying $S(P_m) \subset \P_m$. In particular, $\psi_m \in \P_m$ for each $m=0,\dots,n$. Then the orthonormality of the $\psi_m$ imply that they are simply the continuous Legendre polynomials on $[0,1]$, again up to a sign convention. Note that even though the $\varphi_m$ depend on $n$, the $\psi_m$ do not.

\paragraph{Quantifying redundancy of the Bernstein system.}
In frame theory, the redundancy of a finite frame is often defined as the ratio of the number of frame vectors to the dimension of their span. Clearly this definition is insufficient for the purpose of differentiating between all the ways in which $n$ vectors can span a $d$ dimensional space, let alone for handling infinite dimensional spaces. A particular deficiency of working with the algebraic span is that it treats all bases the same, whether orthonormal or ill-conditioned. Unfortunately a more suitable quantitative definition of redundancy has been elusive. In the case of unit norm tight frames, the ratio $n/d$ is equal to the frame constant, i.e. the unique eigenvalue of the frame operator $SS^*$ (equivalently, the unique non-zero eigenvalue of the Gram matrix), and for near-tight frames, the frame bounds given by the smallest and largest eigenvalues of the frame operator may continue to serve as a rough substitute of redundancy. However, when the eigenvalues are widely dispersed without a well-defined gap, the frame bounds seem to lose their significance. 

As a basis of $\P_n$, the Bernstein system $\B_n$ would not be considered redundant by the classical definition. However, it is ill-conditioned; in fact, we have 
$$\frac{\sigma_0}{\sigma_n} = \sqrt{\frac{\binom{2n+1}{n+1}}{n+1}}$$ 
which is exponentially large in $n$.
Hence the relevant question is: What is the ``effective dimension'' of the span of $\B_n$? One possible answer to this question comes from numerical analysis via the notion of ``numerical rank'' (see e.g. \cite[Ch. 5.4.1]{golub2013matrix}).  This quantity is tied to a given tolerance threshold $\epsilon$  measuring which singular values (and therefore which subspaces) should count as significant. We inspect the singular values against the top singular value $\sigma_0$ and define 
$$
d_n(\epsilon)
:= \max\{0\leq m \leq n \colon \sigma_{m} \geq \epsilon \sigma_0\}
$$
where $\epsilon \in (0,1)$ is a fixed small parameter. With this definition, if we truncate the singular value decomposition of $S$ to define
$$
\tilde S_\epsilon u := \sum_{\sigma_m \geq \epsilon \sigma_0} \sigma_m \langle u, \varphi_m\rangle \psi_m, 
$$
then a straightforward consequence is that
$$ \frac{\|S - \tilde S_\epsilon\|_{\rm op}}{\|S\|_{\rm op}} \leq \epsilon.$$
With this, we can argue that the image of the unit ball under $S$ is well approximated by an ellipsoid within $\mathrm{span}(\psi_0,\dots,\psi_{d_n(\epsilon)})= \P_{d_n(\epsilon)}$. Hence we define the {\em $\epsilon$-redundancy} of $\B_n$ to be $n/d_n(\epsilon)$. Note that this definition can easily be applied to systems other than $\B_n$.

The next task is to obtain asymptotics for $d_n(\epsilon)$. As can be seen from the formula derived in Theorem \ref{thm:Gammaeig}, the singular values of $S$ do not possess any obvious gap. It turns out that the exponentially large condition number of $S$ is highly misleading because the singular values of $S$ do not decay like an exponential. If $\sigma_m$ actually decayed like $\exp(-cm)$, we would find that $d_n(\epsilon) \sim c^{-1} \log (1/\epsilon)$, which is independent of $n$. Instead, we will show below that 
$$ d_n(\epsilon) = C_\epsilon \sqrt{n} (1 + o(1))\mbox{ as }n \to \infty,$$
where $C_\epsilon := \sqrt{2 \log (1/ \epsilon)}$. This result is consistent with the fact that the error of best approximation using $\B_n$ (with bounded coefficients) behaves as if $n$ is replaced by $\sqrt{n}$. In addition, note that the $\epsilon$-redundancy of $\B_n$ given by $n/d_n(\epsilon)$ also behaves like $\sqrt{n}$. This is also consistent with our result that the error bound of the $r$-th order $\Sigma\Delta$ quantization method behaves as $n^{-r/2}$.

\begin{theorem}
	\label{thm:Gammaeig2}
	For all $\epsilon\in (0,1)$, 
	\begin{equation}
		\label{asmptotic-dn}
		\lim_{n \to \infty} \frac{d_n(\epsilon)}{ \sqrt{2n \log (1/\epsilon)}}  = 1.	
	\end{equation}
\end{theorem}

\begin{proof}
	The phrase ``for all sufficiently large $n$'' will occur several times below and will be abbreviated by ``f.a.s.l. $n$.''	Suppose $\epsilon \in (0,1)$ is fixed. Define
	\begin{equation}
		\label{mmprime}
		m_n:= \left \lfloor \sqrt{2(n{+}1) \log (1/\epsilon)} \right \rfloor -1,~~ \mbox{ and }~~ m'_n:= \left \lceil \sqrt{2(n{+}1) \log (1/\epsilon)} \right \rceil.
	\end{equation}
	It is clear that 
	\begin{equation} 
		\label{sandwich-1} 
		\lim_{n \to \infty} \frac{m_n}{ \sqrt{2n \log (1/\epsilon)}} =
		\lim_{n \to \infty} \frac{m'_n}{ \sqrt{2n \log (1/\epsilon)}}  = 1.
	\end{equation}
	As a trivial consequence, $m_n$ and $m'_n$ are in $\{0,\dots,n\}$ f.a.s.l. $n$, so we can meaningfully refer to $\sigma_{m_n}$ and $\sigma_{m'_n}$. The proof will be based on showing that 
	\begin{equation}
		\label{sandwich-alpha} 
		\sigma_{m_n} > \epsilon \sigma_0 > \sigma_{m'_n} ~~
		\mbox{ f.a.s.l. $n$} 
	\end{equation}
	which implies $m_n \leq d_n(\epsilon) < m'_n$ f.a.s.l. $n$, and therefore \eqref{asmptotic-dn} via \eqref{sandwich-1}.

	It will be more convenient to work with the increasing sequence $\alpha_m:= \log ( \lambda_0/\lambda_m)$, $m=0,\dots,n$.
	By Theorem \ref{thm:Gammaeig}, we have 
	$$
	\alpha_m
	= \log \Big( \frac{(n{+}m{+}1)_{m+1}}{(n{+}1)_{m+1}}\Big)
	= \sum_{k=0}^m \left( \log \Big(1{+}\frac{k}{n{+}1}\Big)-\log\Big( 1{-}\frac{k}{n{+}1}\Big) \right) 
	= \sum_{k=0}^n g\Big(\frac{k}{n{+}1} \Big),
	$$
	where $g(t):=\log(1+t) - \log(1-t)$, $|t| < 1$. We have $g(0)=0$, $g'(0)=2$, $g''(0)=0$, and $|g'''(t)| \leq 18$ for $|t| \leq 1/2$, so that $|g(t) - 2t| \leq  3|t|^3$ for $|t| \leq 1/2$ via Taylor's theorem.
	Hence, for all $m \leq \frac{n+1}{2}$, we have
	\begin{equation}
		\label{cubic-bound}
		\left | \alpha_m
		- \frac{m(m{+}1)}{n{+}1} \right|  =
		\left | \sum_{k=0}^m g\Big(\frac{k}{n{+}1} \Big)
		- \sum_{k=0}^m \frac{2k}{n{+}1} \right|  \leq 
		3 \sum_{k=0}^m \Big(\frac{k}{n{+}1} \Big)^3
		< \frac{(m{+}1)^4}{(n{+}1)^3}.
	\end{equation}
	We will utilize \eqref{cubic-bound} to estimate $\alpha_{m_n}$ and $\alpha_{m'_n}$. (Note that
	\eqref{sandwich-1} also implies that $0\leq m_n < m'_n \leq \frac{n{+}1}{2}$ f.a.s.l. $n$.) The definition \eqref{mmprime} readily implies
	\begin{equation}
	\label{sandwich-log}
	 \frac{(m_n{+}1)^2}{n{+}1} \leq \log(1/\epsilon^2) \leq \frac{(m'_n)^2}{n{+}1},
	\end{equation}
	and rearranging these inequalities, we have
	\begin{equation}
	\label{sandwich-log-2}
	\frac{m_n(m_n{+}1)}{n{+}1} \leq 
	\log(1/\epsilon^2) - \frac{m_n{+}1}{n{+}1}
	~~\mbox{ and }~~
	\frac{m'_n(m'_n{+}1)}{n{+}1} \geq 
	\log(1/\epsilon^2) + \frac{m'_n}{n{+}1}.	
	\end{equation}
	The estimate \eqref{cubic-bound} for $m=m_n$ with \eqref{sandwich-log} and \eqref{sandwich-log-2} yields
	\begin{eqnarray}
		\alpha_{m_n} & \leq & \frac{m_n(m_n{+}1)}{n{+}1} + \frac{(m_n{+}1)^4}{(n{+}1)^3}
		\nonumber \\
		& \leq & \log(1/\epsilon^2) - \frac{m_n{+}1}{n{+}1} + \frac{\log^2(1/\epsilon^2)}{n{+}1}
		\nonumber \\
		& < & \log(1/\epsilon^2) ~~\mbox{ f.a.s.l. $n$}
	\end{eqnarray}
	since $m_n \to \infty$.
	Similarly, we also have
	\begin{eqnarray}
	\alpha_{m'_n} & \geq &\frac{m'_n(m'_n{+}1)}{n{+}1} - \frac{(m'_n{+}1)^4}{(n{+}1)^3} \nonumber \\
	& \geq &
	\log(1/\epsilon^2) + \frac{m'_n}{n{+}1} - \frac{2 \log^2(1/\epsilon^2)}{n{+}1} ~~\mbox{ f.a.s.l. $n$}, \nonumber \\
	& > & \log(1/\epsilon^2) ~~\mbox{ f.a.s.l. $n$},
	\end{eqnarray} 
	where in the second last step we used the observation that 
	$(m'_n+1)^4 \leq 2(m_n+1)^4$ f.a.s.l. $n$.
	Hence we have shown
	$$\alpha_{m_n} < \log(1/\epsilon^2) < \alpha_{m'_n} ~~\mbox{ f.a.s.l. $n$},	$$
	which is equivalent to \eqref{sandwich-alpha}.
\end{proof}

\section*{Appendix B: One-bit neural networks}

Universal approximation properties of feedforward neural networks are well-known (see \cite{cybenko1989approximation} as well as the more recent treatments including \cite{yarotsky2017error,shaham2018provable,lu2021deep,bolcskei2019optimal,daubechies2021nonlinear} and the extensive survey \cite{devore2021neural}). One of the motivations of this work has been to investigate the approximation potential of feedforward neural networks with the additional constraint that their parameters are coarsely quantized, at the extreme using only two values, hence the term ``one-bit''. Quantized networks have been studied from a variety of perspectives (see e.g. \cite{courbariaux2015binaryconnect,guo2018survey,ashbrock2021stochastic,lybrand2021greedy}). As far as we know, a universal approximation theorem using one-bit, or even fixed-precision multi-bit networks was missing. In this appendix we will show how this can be done using our results on polynomial approximation in the Bernstein system. In particular, employing the quadratic activation unit $s(x) := \frac{1}{2}x^2$, we will show that it is possible to implement the polynomial approximations of this paper using standard feedforward neural networks whose weights are chosen from $\{\pm 1\}$ only. Due to the scope and focus of this paper, we will limit our discussion to univariate functions.
A much more comprehensive study of these networks covering multivariate functions, ReLU networks, and information theoretic considerations can be found in our separate manuscript \cite{onebitNN}. 

A feedforward neural network can be described in a variety of general formulations (see e.g. \cite{devore2021neural}). The architecture of the network is determined by a directed acyclic graph. The vertices of this graph (called nodes of the network) consist of input, output, and hidden vertices. Vertices are assigned variables (independent or dependent depending upon whether the vertex is an input or not) and edges are assigned weights. Every node which is not an input node computes a linear combination of the variables assigned to the nodes that are connected to it (noting the directedness of the graph) using the weights associated to the connecting edges, followed by (except for the output nodes) an application of a given {\em activation function} $\rho$, possibly subject to a shift of its argument (called the bias). An important class of networks are {\em layered}, meaning that the nodes of the network can be partitioned into subsets (layers) such that the nodes in each layer have incoming edges only from one other (the previous) layer and outgoing edges into another (the next) layer.

Having introduced the basic terminology of feedforward neural networks, let us now turn the polynomial approximation method of this paper into a neural network approximation with $\pm 1$ weights.
There are two critical ingredients in our construction of these networks. The first one is, of course, the fact that we are able to approximate functions $f$ (where $\|f\|_\infty \leq 1$) by $\pm 1$-linear combinations of the elements of the Bernstein system $\B_n$. The second critical ingredient is that the basis polynomials $p_{n,k}$ have a rich hierarchical structure that enables them to be computed recursively, and with few simple arithmetic operations. Indeed, the $p_{n,k}$ satisfy the elementary recurrence relation
\begin{equation}
	p_{n,k}(x) = \left \{
	\begin{array}{ll}
		(1-x)p_{n-1,0}(x), & k=0,\\
		x p_{n-1,k-1}(x)+ (1-x)p_{n-1,k}(x), & 0 < k < n, \\
		x p_{n-1,n-1}(x), & k=n,
	\end{array}
	\right.
\end{equation}
which follows readily from the simple combinatorial identity $\binom{n}{k} = \binom{n-1}{k-1} + \binom{n-1}{k}$. (The endpoint cases of $k=0$ and $k=n$ can actually be eliminated if we use the earlier convention  $p_{n,j}(x) = 0$ if $j < 0$ or $j>n$.)
Since this combinatorial identity is the basis of the Pascal triangle, we will refer to the resulting tree structure of the Bernstein basis polynomials as the {\em Pascal-Bernstein tree}, which is depicted in the top triangular portion of the graph in Figure \ref{fig:PB}. 
The bottom portion of this graph shows how these basis polynomials are combined with weights $\sigma_k \in \{\pm 1\}$, $k=0,\dots,n$. The red edges  of the tree correspond to multiplication by $1-x$ and the blue edges correspond to multiplication by $x$. Note that this schematic diagram is not a proper feedforward neural network yet because the weights depend on the input $x$. However, it is readily implementable by a {\em sum-product network} \cite{poon2011sum} where the nodes either multiply or compute a weighted sum their inputs. In our case, the weights consist of ${\pm 1}$ only. The input to the network could be $x$ and $\bar x:=1-x$.
\begin{figure}[th]
	\centerline{\includegraphics[width=0.5\textwidth]{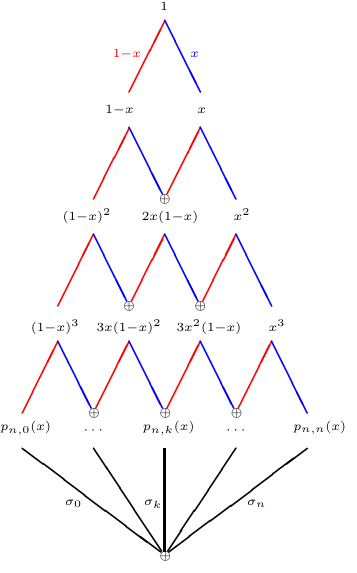}}
	\caption{Schematic diagram of the Pascal-Bernstein tree and the associated one-bit polynomial approximations in the Bernstein basis. Red edges of the tree correspond to multiplication by $1-x$ and blue edges by $x$.} 
	\label{fig:PB} 
\end{figure}
Regardless of the nature of this network (or proto-network), the Pascal-Bernstein tree portion of the algorithm is universal in the sense that it is the same for every function $f$ to be approximated. All of the information about how $f$ is approximated is encoded in the weights $(\sigma_k)_0^n$ connecting the bottom layer to the final output.

In order to turn this schematic diagram into a proper feedforward neural network, it suffices to convert its multiplications (by $1-x$ and $x$) into a standard neural network operation. We show below that this can be done using $\pm1$ weights in the simplest way if the activation unit is given by the quadratic function $s(u):=\frac{1}{2}u^2$. Indeed, with this unit, the multiplication of any two quantities $a$ and $b$ can be implemented simply as
$$ ab = s(a+b) - s(a) - s(b).$$
Then, for any $m=1,\dots,n$, the quantities $p_{m{-}1,j}(x)$, $j=0,\dots,m-1$, associated with the vertices in the $m$th layer of this tree satisfy the relations
\begin{equation}\label{recurse-Bernstein-network}
	p_{m,k}(x) = \left \{
	\begin{array}{ll}
		s(1-x+p_{m-1,0}(x)) - s(1-x) - s(p_{m-1,0}(x)), & k=0,\\
		\begin{array}{l}\Big (s(x+p_{m-1,k-1}(x)) - s(x) - s(p_{m-1,k-1}(x))  \\
		+\, s(1-x+p_{m-1,k}(x)) - s(1-x) - s(p_{m-1,k}(x))\Big)
		\end{array}&  0 < k < m, \\
		s(x+p_{m-1,m-1}(x)) - s(x) - s(p_{m-1,m-1}(x)), & k=m.
	\end{array}
	\right.
\end{equation}
This representation shows that the $p_{m,k}(x)$ will not actually correspond to the physical outputs of the nodes of our neural network, but rather to certain intermediate $\pm 1$ linear combinations of the outputs of up to $6$ of its nodes. If $m < n$, we add $x$ or $1-x$ to these before passing them onto subsequent nodes. If $m=n$, we take a final $\pm 1$ linear combination using the coefficients $\sigma_k$. 

Let us describe the nodes of our network more precisely. For each $m=0,\dots,n-1$, there will be a layer whose nodes output the functions $X_{m,k}$, $Y_{m,k}$, $Z_{m,k}$, $k=0,\dots,m$, where we would like
$$ X_{m,k}(x) = s(p_{m,k}(x)), ~~~
Y_{m,k}(x) = s(1-x+p_{m,k}(x)), ~~~
Z_{m,k}(x) = s(x+p_{m,k}(x)).$$
Therefore, we set $U:= s(1-x)$ and $V:=s(x)$ and define these functions for $m=1,\dots,n-1$ via the recurrence 
\begin{eqnarray}
	X_{m,k} & := & \left \{
	\begin{array}{ll}
		s(Y_{m-1,0}-U-X_{m-1,0}), & k=0,\\
		s(Z_{m-1,k-1}-V-X_{m-1,k-1}+Y_{m-1,k}-U-X_{m-1,k}),&  0 < k < m, \\
		s(Z_{m-1,m-1}-V-X_{m-1,m-1}), & k=m,
	\end{array}
	\right. \nonumber \\
	Y_{m,k} & := & \left \{
	\begin{array}{ll}
		s(1-x+Y_{m-1,0}-U-X_{m-1,0}), & k=0,\\
		s(1-x+Z_{m-1,k-1}-V-X_{m-1,k-1}+Y_{m-1,k}-U-X_{m-1,k}),&  0 < k < m, \\
		s(1-x+Z_{m-1,m-1}-V-X_{m-1,m-1}), & k=m,
	\end{array}
	\right. \nonumber \\
	Z_{m,k} & := & \left \{
	\begin{array}{ll}
		s(x+Y_{m-1,0}-U-X_{m-1,0}), & k=0,\\
		s(x+Z_{m-1,k-1}-V-X_{m-1,k-1}+Y_{m-1,k}-U-X_{m-1,k}),&  0 < k < m, \\
		s(x+Z_{m-1,m-1}-V-X_{m-1,m-1}), & k=m.
	\end{array}
	\right. \nonumber 	
\end{eqnarray}
We also set $X_{0,0} := s(1)$, $Y_{0,0} := s(1-x+1)$, $Z_{0,0}:=s(x+1)$. 

The final output of the network is the function $\displaystyle f_{\mathrm{NN}}(x)$ given by $\sum_k \sigma_k p_{n,k}(x)$. The $p_{n,k}(x)$ are still available indirectly through \eqref{recurse-Bernstein-network} for $m=n$, so we define 
\begin{eqnarray}
	f_\mathrm{NN} & := & \sigma_0 Y_{n-1,0} - \sigma_0 U - \sigma_0 X_{n-1,0}  \nonumber \\
	& & + \,\sum_{k=1}^{n-1} (\sigma_k Z_{n-1,k-1} - \sigma_k V - \sigma_k X_{n-1,k-1} + \sigma_k Y_{n-1,k} - \sigma_k U - \sigma_k X_{n-1,k}) \nonumber \\
	& & + \,\sigma_{n-1} Z_{n-1,n-1} - \sigma_{n-1} U - \sigma_{n-1} X_{n-1,n-1}.	
	\label{output-layer}
\end{eqnarray}

We note that this expression contains two copies of the
$X_{n-1,j}$, $j=0,\dots,n-1$, each carrying a $\pm 1$ weight. If these weights were to be combined, then we would get a new weight in the set $\{-2,0,2\}$. This inconsistency can easily be removed by duplicating the nodes that produce $X_{n-1,j}$. Of course, the same comment applies to $U$ and $V$ which would need to be copied $n$ times. A copy of a node is easily created using the identity
$$ a = s(a+1) - s(a) - s(1).$$

The appearance of $U$ and $V$ in each step of the recurrence means that the network is not completely layered, containing some ``skip'' connections. However, the copying mechanism above can also be used as a repeater to remove these skip connections; see Figure \ref{fig:repeater}. All of these additional operations can also be avoided by means of {\em special networks} (e.g. \cite{daubechies2021nonlinear}) where it is permissible to create certain channels to push forward input values or other intermediate computations, or by allowing for more than one type of activation function to be used at the nodes (e.g. the identity function).

\begin{figure}[th]
	\centerline{\includegraphics[width=0.5\textwidth]{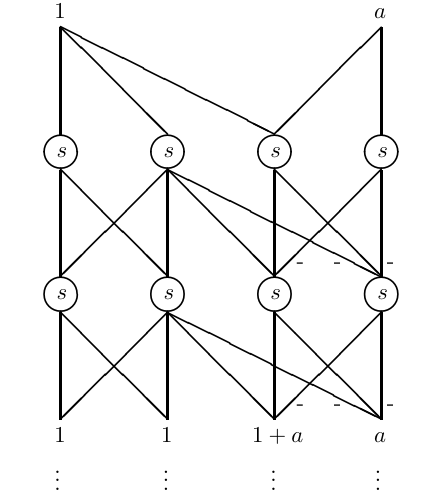}}
	\caption{Demonstration of how the output $a$ of any node or the input $a\in\{x,1\}$ of the network can be repeated to be turned an input to any node in a subsequent layer. Two layers of repetitions are shown.} 
	\label{fig:repeater} 
\end{figure}

Finally, we note that our network takes as input both $x$ and $1$ (or alternatively, $x$ and $1-x$). This choice has allowed us to avoid the use of any bias values associated to the activation function.

\end{document}